\documentclass[10pt,conference,letter]{IEEEtran}
\usepackage{epsfig}
\usepackage{amsmath}
\usepackage{algorithm}
\usepackage{algorithmicx}
\usepackage{algpseudocode}
\newtheorem{theorem}{Theorem}[section]

\newtheorem{remarks}{Remarks}[section]

\begin{document}

\title{An Information Theoretic Point of View to Contention Resolution
}

\begin{abstract}
We consider a slotted wireless network in an infrastructure setup
with a base station (or an access point) and $N$ users.
The wireless channel gain between the base station and the users is
assumed to be i.i.d., and the base
station seeks to schedule the user with the highest channel gain in every slot (opportunistic scheduling).
We assume that the
identity of the user with the highest channel gain is resolved using a series of contention slots and with feedback from the base station.
In this setup, we formulate the
contention resolution problem for opportunistic scheduling as
identifying a random threshold (channel gain) that separates the best channel
from the other samples.
We show that the average delay to resolve contention is related to the entropy of the
random threshold.

We illustrate
our formulation by studying the opportunistic splitting algorithm
(OSA) for i.i.d. wireless channel \cite{qin-berry-osa}. We note that
the thresholds of OSA correspond to
a maximal probability allocation scheme. We
conjecture that maximal probability allocation is an entropy
minimizing strategy and a delay minimizing strategy for i.i.d. wireless channel.
Finally, we discuss the applicability
of this framework for few other network scenarios.
\end{abstract}

\maketitle

\section{Introduction}
The advancements in the physical layer technology has enabled cellular networks (e.g., 3G and 4G deployments like Mobile WiMAX, LTE Advanced) and WLANs (e.g., IEEE 802.11n) support hundreds of megabits per second.
However, with more and more users now accessing the Internet using wireless as the last mile, there is a continuous necessity to judiciously use the available network resources.
Cross-layer strategies have become extremely helpful in supporting the ever increasing demand for bandwidth and stringent QoS.
Opportunistic scheduling and multiuser diversity (see \cite{liu-etal-os})
is one such popular cross-layer technique recommended in current cellular standards and in ad hoc deployments for increasing the available network capacity. Unlike the wired channel, the wireless channel will always be constrained by fading and interference. Multiuser diversity enhances the network performance by wisely scheduling the users when their relative channel conditions are better. Opportunistic scheduling is known to significantly improve the network performance, especially for elastic traffic with loose delay constraints.

Opportunistic scheduling involves learning the channel state information of the contending users and scheduling the user with a relatively better channel. Centralized schemes like polling incur a lot of overhead and may not scale well with the number of users.
For such schemes, the rate region of the channel and the set of feasible QoS are well known (see e.g., \cite{neely-etal-powerrouting}).
The performance of the system with partial channel state information was studied in \cite{gopalan-etal-partialcsi}.
There is a lot of interest in developing distributed and semi-distributed algorithms for opportunistic scheduling. One popular technique has been to adjust the backoff parameters of the nodes based on their instantaneous channel gain. A number of works have studied the optimal performance and the achievable throughput of such strategies (see e.g., \cite{qin-berry-multiuser}).
In \cite{qin-berry-osa}, the authors propose a splitting algorithm that resolves contention with feedback from the base station. The distributed strategies incur losses due to collisions but are known to very efficient especially for networks with a large number of users.

In this work, we are interested in the contention resolution problem of resolving the identity of the user with the highest channel gain. We formulate the contention resolution problem for opportunistic scheduling as identifying a random threshold (channel gain) that separates the best channel from the other samples. We show that the average delay to resolve contention is related to the entropy of the threshold random variable. We illustrate our formulation by studying the opportunistic splitting algorithm \cite{qin-berry-osa}. We show that OSA is a maximal probability allocation scheme and we conjecture that MPA is an entropy minimizing strategy and a delay minimizing strategy as well.
In this work, we have studied opportunistic scheduling for $N$ users with i.i.d. channel gains.
We believe that our formulation of contention resolution as a source code can help develop optimal strategies for a variety of other network scenarios as well.

\subsection{Related Literature}

The idea of splitting with ternary feedback was originally proposed for scheduling users in Aloha type networks (see \cite{gallager-conflictresolution}).
In \cite{arrow-etal-binaryquestions}, Arrow et al., study a problem of resolving the user with the highest sample value with binary type questions. The optimal strategy was studied when accurate feedback of the number of contending users involved in every slot was available.
The near optimality of greedy strategies (like MPA studied in Section~\ref{sec:osa_source_coding}) was also discussed in \cite{arrow-etal-binaryquestions}.
In \cite{anantharam-varaiya-conflictresolution}, Anantharam and Varaiya prove the optimality of
binary type questions to minimize the average delay in \cite{arrow-etal-binaryquestions}. The performance of binary type questions in the presence of ternary feedback was first reported in \cite{cohen-etal-ternaryfeedback}. The optimal thresholds were obtained and the relevance to opportunistic scheduling was discussed.

In \cite{qin-berry-osa}, Qin and Berry study splitting with ternary feedback for opportunistic scheduling for i.i.d. wireless channel. We have briefly described the algorithm in Section~\ref{sec:osa}; we motivate our formulation of contention resolution as a source coding problem by studying the opportunistic splitting algorithm presented in \cite{qin-berry-osa}.
Splitting algorithms have been studied for other network and channel scenarios as well.
In \cite{kessler-sidi-splittingnoisy}, Kessler and Sidi study splitting algorithms for noisy channel feedback. In
\cite{qin-berry-osafairness}, Qin and Berry report the performance of splitting for different notions of fairness.
In \cite{yu-giannakis-sic}, Yu and Giannakis study the performance of splitting with successive interference cancellation in a tree algorithm.
In this work, we restrict to i.i.d. wireless channel under ideal channel assumptions; our aim is to present an alternate formulation for contention resolution using a source coding framework.

There are number of works concerning distributed opportunistic feedback schemes for wireless systems (see e.g., \cite{tang-etal-opportunisticfeedback}).
In \cite{qin-berry-multiuser}, Qin and Berry proposes a channel aware ALOHA and characterizes its performance.
In \cite{patil-deveciana-reducingfeedback}, Patil and de Veciana discuss about reducing feedback for opportunistic scheduling to support best effort and real time traffic.
In this work, we consider a semi-distributed framework where the base station helps resolve contention with feedback.

\subsection{Outline}
In Section~\ref{sec:network_model}, we describe the network model and the opportunistic resolution problem.
In Section~\ref{sec:osa}, we briefly describe the opportunistic splitting algorithm from \cite{qin-berry-osa} and motivate our formulation.
In Section~\ref{sec:source_coding}, we present contention resolution problem
for opportunistic scheduling as a source coding problem. In Section~\ref{sec:osa_source_coding}, we characterize
OSA using a maximal probability allocation scheme and study its performance. In Section~\ref{sec:two_examples}, we discuss
the applicability of our framework for other network scenarios and in Section~\ref{sec:conclusion}, we conclude the
paper and discuss future work.

\section{Network Model}
\label{sec:network_model}
We consider the downlink wireless channel of a single cell of a cellular data network (or of a single cell WLAN in an infrastructure setup).
A fixed number of users, $N$, share the slotted wireless channel over time.
We assume that the channel gain between the base station and the
wireless users is independent and identically distributed
with a common continuous distribution $F(\cdot)$.
We also assume that the users have knowledge of the common channel distribution $F(\cdot)$ and the number of users in the network, $N$.

Let $(H_{n,1}, H_{n,2}, \cdots, H_{n,N})$ represent the vector channel gain of the users in slot $n$.
We assume that every user $i$ would know its instantaneous channel gain $H_{n,i}$ at the beginning of every slot, but that information is not available with other users in the network, including the base station.
The channel state information $H_{n,i}$ can be made available to the user $i$ by the transmission of a pilot signal by the base station at the beginning of the slot.
The base station seeks to identify and schedule
the user with the highest channel gain in every slot (opportunistic scheduling),
i.e., the base station seeks to schedule
\[ \arg \max_{\{i=1,\cdots,N\}} \{H_{n,1}, H_{n,2}, \cdots, H_{n,N}\} \]
in slot $n$.
Define $X_{n,i} := F(H_{n,i})$, the cumulative distribution value in the slot $n$. Then, the vector $(X_{n,1}, X_{n,2}, \cdots, X_{n,N})$ is i.i.d. Uniform in $[0,1]$ for any channel distribution $F(\cdot)$. Further, the contention resolution problem can equivalently be described as
\[ \arg \max_{\{ i=1,\cdots,N\}} \{X_{n,1}, X_{n,2}, \cdots, X_{n,N}\}\]
Hence, without loss of generality, we will assume that $F(\cdot)$ is Uniform in $[0,1]$ and consider $(X_{n,1}, X_{n,2}, \cdots, X_{n,N})$ as the channel gain variables.

The base station resolves the identity of the user with the highest channel gain by
coordinating the contention resolution process and by providing necessary feedback to aid in the resolution.
We assume that a time slot comprises of $K$ mini slots, where the mini slots are used to resolve the contention.
For example, the users can transmit MAC packets (like RTS/CTS in IEEE 802.11 DCF), possibly with some channel information, to the base station in a minislot and the base station can feed back the state of the contention in that slot.
We assume that the base station feeds back the result of the contention within the minislot and the feedback of the base station is received by all the nodes in the network without any error.
At the end of the contention process, the user that succeeded in the contention is permitted to transmit data in the remainder of the slot.
In this setup, an objective of the base station would be to minimize the average number of minislots required to identify the user with the highest channel gain.

\section{Opportunistic Splitting}
\label{sec:osa}
In this section, we briefly describe a contention resolution strategy, opportunistic splitting algorithm (OSA) from \cite{qin-berry-osa}, for a fixed number of users $N$ and for i.i.d. block fading wireless channel.
Polling for opportunistic scheduling requires $N$ minislots to identify the user with the highest channel gain.
OSA is a distributed medium access control protocol that uses ternary feedback to identify the user with the best channel with a constant overhead.

A time slot is assumed to comprise of a maximum of $K$ minislots which are used for contention resolution.
In every minislot, OSA describes a continuous range in $[0,1]$ (the sample space of the Uniform random variable), $(y_{min}, y_{max}] \subset [0,1]$; only the user(s) whose channel gain values fall within the range will transmit contention resolution packets in the minislot.
At the end of the minislot, every user receives a feedback from the base station of $0$ or $1$ or $e$, indicating if the minislot was idle (no transmission), contained a successful packet transmission or involved an error due to collision, respectively.
If the feedback is $1$, the lone transmitter is declared the winner of the contention and is permitted to transmit data for
the remaining duration of the slot.
If the feedback is $0$ or $e$, then the range is suitably adjusted and the contention resolution process continues until either a success occurs or the time-slot ends.

The following pseudo-code describes the OSA algorithm for a fixed number of users $N$ and for i.i.d. channel gain (see \cite{qin-berry-osa} for more details). In the pseudo-code, $f$ denotes the feedback in a minislot and $k$ is the count of the number of minislots used for contention resolution.

\begin{algorithmic}
\State Initialize: $y_{low} = 0, y_{min} = 1 - \frac{1}{N}, y_{max} = 1$
\State Initialize: $f = 0$ and $k = 1$
\While{($f \neq 1$) and ($k <= K$)}
    \State $f = (0,1,e)$ feedback from $(y_{min}, y_{max}]$
    \If {($f = e$)}
        \State $y_{low} = y_{min}$
        \State $y_{min} = \frac{ (y_{min} + y_{max})}{2}$
    \EndIf
    \If {($f = 0$)}
        \State $y_{max} = y_{min}$
        \If {($y_{low} \neq 0$)}
            \State $y_{min} = \frac{(y_{low} + y_{max})}{2}$
        \Else
            \State $y_{min} = y_{max} (1 - \frac{1}{N})$
        \EndIf
    \EndIf
    \State $k = k + 1$
\EndWhile
\end{algorithmic}

\begin{remarks}
\label{rem:osa}
The key features of the opportunistic splitting algorithm are the following.
\begin{enumerate}
\item OSA aims to maximize the chances of success in every minislot. For example, with $N$ users independently and Uniformly distributed in $[0,1]$, the probability of success (identifying the user with the best channel) in a minislot with the range $(p,1]$, $N p (1 - p)^{N-1}$, is maximized at $p = 1 - \frac{1}{N}$. In fact, OSA begins contention resolution with the range $(1 - \frac{1}{N}, 1]$.
\item When a collision occurs, OSA assumes that the most likely scenario is that two users are involved in the collision, and hence, it updates the threshold from $(y_{min}, y_{max}]$ to $( \frac{y_{min} + y_{max}}{2}, y_{max}]$ (the optimal strategy if there are only two contending users).
\end{enumerate}
\end{remarks}

OSA is an effective contention resolution strategy with the average number of minislots required to resolve contention
known to be less than 2.5070 slots, independent of the number of users and channel gain distribution.

\subsection{Two User Case}
In this section, we will discuss in detail the opportunistic splitting algorithm for the two user case. The example will help us motivate the source coding framework described in the Section~\ref{sec:source_coding}.
Let $N = 2$ and let $(X_1, X_2)$ correspond to the vector channel gain of the two users in a slot. Define $Y_1 := \min\left(X_1, X_2\right)$ and $Y_2 := \max \left( X_1, X_2 \right)$. Then, $(Y_1, Y_2)$ is the ordered pair of the channel gain values where $0 \leq Y_1 \leq Y_2 \leq 1$.

OSA initializes with $y_{low} = 0, y_{min} = \frac{1}{2}$ and $y_{max} = 1$. In the first minislot, only the user(s) with $\frac{1}{2} < X_i$ transmit a control packet. A success (a single transmission) happens in the first minislot iff ($X_1 \leq \frac{1}{2} < X_2$) or ($X_2 \leq \frac{1}{2} < X_1$), i.e., a success happens iff $Y_1 \leq \frac{1}{2} < Y_2$. The probability of the event can easily be computed and is equal to $\frac{1}{2}$. \textbf{Thus, contention is resolved in the first minislot whenever $0 \leq Y_1 \leq \frac{1}{2} < Y_2 \leq 1$ and the probability of the event is $\frac{1}{2}$; the threshold that resolves the contention successfully for the set $\{ (Y_1,Y_2) : 0 \leq Y_1 \leq \frac{1}{2} < Y_2 \leq 1\}$ is $\frac{1}{2}$ and the base station feeds back a $1$ in this case}. In the first minislot, an error due to collision occurs iff $\frac{1}{2} < Y_1 \leq Y_2$ and the slot is left idle iff $Y_1 \leq Y_2 \leq \frac{1}{2}$. Suppose that the feedback in the first minislot is $e$. Then, OSA updates the variables as $y_{low} = \frac{1}{2}, y_{min} = \frac{3}{4}$ and $y_{max} = 1$. In the second minislot, only the user(s) with $\frac{3}{4} < X_i$ transmit a control packet. A success happens now iff $Y_1 \leq \frac{3}{4} < Y_2$ and the conditional probability of the event (conditioned upon a collision in the first minislot) is $\frac{1}{2}$. \textbf{Thus, contention is resolved in the second minislot whenever $\frac{1}{2} < Y_1 \leq \frac{3}{4} < Y_2 \leq 1$ and the probability of the event is $\frac{1}{8}$; the threshold that resolves the contention successfully for the set $\{ (Y_1,Y_2): \frac{1}{2} < Y_1 \leq \frac{3}{4} < Y_2 \leq 1\}$ is $\frac{3}{4}$ and the base station feeds back a $e1$ in the first two minislots}.

\begin{table}
\renewcommand{\arraystretch}{1.5}
\centering
\begin{tabular}{|c|c|c|c|c|}
\hline
Events & Threshold & Feedback & Prob \\
\hline
$0 \leq Y_1 \leq \frac{1}{2} < Y_2 \leq 1$ & $\frac{1}{2}$ &  1 & $\frac{1}{2}$ \\
$\frac{1}{2} < Y_1 \leq \frac{3}{4} < Y_2 \leq 1$ & $\frac{3}{4}$ &  e1 & $\frac{1}{8}$ \\
$0 \leq Y_1 \leq \frac{1}{4} < Y_2 \leq \frac{1}{2}$ & $\frac{1}{4}$ & 01 & $\frac{1}{8}$ \\
$\frac{3}{4} < Y_1 \leq \frac{7}{8} < Y_2 \leq 1$ & $\frac{7}{8}$ & ee1 & $\frac{1}{32}$ \\

$\frac{1}{2} < Y_1 \leq \frac{5}{8} < Y_2 \leq \frac{3}{4}$ & $\frac{5}{8}$ & e01 & $\frac{1}{32}$ \\
$\frac{1}{4} < Y_1 \leq \frac{3}{8} < Y_2 \leq \frac{1}{2}$ & $\frac{3}{8}$ & 0e1 & $\frac{1}{32}$ \\
$0 \leq Y_1 \leq \frac{1}{8} < Y_2 \leq \frac{1}{4}$ & $\frac{1}{8}$ & 001 & $\frac{1}{32}$ \\
\vdots & \vdots & \vdots & \vdots \\
\hline
\end{tabular}
\caption{The probability distribution on the threshold/feedback corresponding to OSA for $N = 2$ users.}
\label{tab:osa_two_users}
\end{table}

In Table~\ref{tab:osa_two_users}, we have listed sets of ordered two tuples along with the threshold ($y_{min}$) for OSA that resolves the set.
The feedback from the base station corresponding to the threshold (equivalently, the set) and the probability of the threshold (equivalently, the feedback) is also listed in the table.

\begin{remarks}
\label{rem:osa_2}
We make the following observations from the Table~\ref{tab:osa_two_users}.
\begin{enumerate}
\item The threshold ($y_{min}$) that resolves $(Y_1,Y_2)$
is always such that $Y_1 \leq y_{min} < Y_2$, i.e.,
OSA resolves contention by identifying a threshold between the user channel gains.
The threshold is fed back to the users in ternary alphabet $(0,1,e)$.
The lone user with a channel gain strictly greater than the threshold value fed back by the base station would learn about its successful contention resolution and the other users would refrain from transmitting any further in the slot.
\item The feedback for a threshold $y_{min}$ is, in fact, the binary expansion of $y_{min}$ (when feedback $e$ and feedback $1$ is mapped to $1$ and feedback $0$ is mapped to $0$). The feedback $1$ is equivalent to feedback $e$ followed by an EoC (end of contention) in this case.
\item The thresholds that resolve contention for OSA form a countable set with a valid probability distribution (the probabilities sum up to $1$).
\item The average delay to resolve contention is equal to the average length of the feedback, which is a function of the probability distribution of the threshold random variable.
The probability distribution is a function only of the contention resolution algorithm (for the i.i.d. case).
An optimal choice of the thresholds can minimize the average description length of the feedback and the delay to resolve contention.
\end{enumerate}
\end{remarks}
In Section~\ref{sec:source_coding}, we will propose a general framework for contention resolution for opportunistic scheduling motivated by the above observations.

\section{A Source Coding Problem}
\label{sec:source_coding}
In this section, we will formulate contention resolution for opportunistic scheduling with ternary feedback as identifying a random threshold (channel gain) that separates the best channel from the other samples. Let $(X_1, X_2, \cdots, X_N)$ correspond to the vector of i.i.d. channel gain values in a slot and let
$(Y_1, Y_2, \cdots, Y_N)$ be the ordered N-tuple of channel gain values of the $N$ users such that $0 \leq Y_1 \leq Y_2 \leq \cdots \leq Y_{N-1} \leq Y_N \leq 1$. The base station seeks to identify $\arg \max_{\{i=1,\cdots,N\}} \{ X_1, \cdots, X_N\}$, or, equivalently, $\arg \{Y_N\}$ in the slot.
We aim to resolve the contention by identifying a threshold $Y$ such that $Y_{N-1} \leq Y < Y_N$; the base station will feedback the threshold $Y$ using ternary alphabet of $0,1$ and $e$ which aids in resolving the contention.
$Y_{N-1}$ and $Y_N$ are random variables, and hence, the threshold $Y$ will also be a random variable. Obviously, the uncertainty in $Y$ would be a measure of the average description length of the threshold/feedback.

Let ${\mathsf C}: [0,1] \times [0,1] \rightarrow [0,1]$ be a code (an allocation),
which assigns for every
2-tuple $(Y_{N-1},Y_N)$ an element $Y := {\mathsf C}(Y_{N-1},Y_N) \in [0,1]$, such that $Y_{N-1} \leq Y < Y_N$.
Let $Y$ have a discrete distribution, i.e., let there exist a set $\Omega_Y = \{y_1, y_2, \cdots \}$ and a set of probabilities $\{p_{y_1}, p_{y_2}, \cdots\}$ such that $\sum_{i=1} p_{y_i} = 1$, and $p_{y_i} := {\mathsf Pr}(Y = y_i)$.
Then, the entropy of the random variable $Y$ (equivalently, the code ${\mathsf C}$) is defined as
\[ - \sum_{i=1}^{\infty} p_{y_i} \log_2(p_{y_i}) \]
Clearly, the entropy would approximate the average length of the feedback required for a contention resolution algorithm that resolves a two tuple $(Y_{N-1},Y_N)$ with threshold ${\mathsf C}(Y_{N-1},Y_N)$.

The code ${\mathsf C}(,)$ can, in general, take a continuous sample space, all of $[0,1]$ and a useful description of entropy may not be possible in such cases. For continuous distribution $F(\cdot)$, ${\mathsf Pr}(Y_{N-1} \neq Y_N) = 1$, and for every $(Y_{N-1}, Y_N)$ such that $Y_{N-1} \neq Y_N$, there is some rational $Q$ such that $Y_{N-1} \leq Q < Y_N$. Hence, we can always identify a code with a countable sample space for any continuous $F(\cdot)$ and define its entropy. Further, the feedback from the base station for any contention resolution algorithm is a finite sequence in ternary alphabet.
Hence, we will always seek a code with a discrete distribution for $Y$.
In such a framework, our objective could be to identify the code with the minimum entropy.

\begin{remarks}
\end{remarks}
\begin{enumerate}
    \item The maximal probability allocation scheme of OSA (see Section~\ref{sec:osa_source_coding})
provides us a discrete distribution for the random threshold $Y$ with a finite entropy.
    \item The ternary description of the threshold does not use all the alphabets completely. For example, the alphabet $1$ appears only at the end of every code word (EoC). Further, the code is non-singular but need not be uniquely decodable as the codes are decoded one at a time. Hence, the entropy of the threshold need not exactly measure the average feedback length (and the average delay).
    \item We have assumed that the code ${\mathsf C}$ is a function only of the two tuple $(Y_{N-1}, Y_N)$. For correlated wireless channels and for arbitrary feedback schemes, we may need to consider ${\mathsf C}$ as a function of the N-tuple $(Y_{1},\cdots,Y_N)$.
\end{enumerate}

\section{OSA as a Source Code}
\label{sec:osa_source_coding}
The opportunistic splitting algorithm with ternary feedback identifies a threshold $Y$ for every $N$ tuple ($Y_1,\cdots,Y_N)$ such that $Y_{N-1} \leq Y < Y_N$. OSA chooses the thresholds in a minislot such that the probability of success is maximized in the minislot. The following pseudo-code describes the code ${\mathsf C}: (Y_{N-1}, Y_N) \rightarrow Y$ corresponding to OSA.


\begin{algorithmic}
\State $\Omega_{()} := \{ (Y_{N-1},Y_N) : 0 \leq Y_{N-1} \leq Y_N \leq 1\}$
\State Initialize $k = 1$
\Repeat
    \State a) Define $y_k$ as
{\small    \[y_k := \arg \max_{y \in [0,1]} {\mathsf Pr}(Y_{N-1} \leq y < Y_N | (Y_{N-1}, Y_N) \in \Omega_{()})\] }
    \State b) Define $\Omega_{y_k}$ as
{\small   \[ \Omega_{y_k} = \{(Y_{N-1},Y_N) : (Y_{N-1},Y_N) \in \Omega_{()}, Y_{N-1} \leq y_k < Y_N\} \] }
    \State c) Update $\Omega_{()}$ as
{   \[ \Omega_{()} = \Omega_{()} \backslash \Omega_{y_k} \] }
    \State $k = k + 1$
\Until{$\Omega_{()} \neq$ NULL}
\end{algorithmic}

We define $\Omega_{()}$ as the set of all $2$ tuples $(Y_{N-1},Y_N)$. The code begins with identifying a threshold $y_1$ that maximizes the probability of success in $\Omega_{()}$. Every two tuple in $\Omega_{()}$ that contains the threshold $y_1$ is assigned to be resolved by the threshold; we define the above set as $\Omega_{y_1}$, the set resolved by the threshold $y_1$. The set $\Omega_{y_1}$ is now removed from $\Omega_{()}$ and the procedure continues. Define $\Omega_Y$ as the set of all thresholds defined by the above pseudo-code. From the construction of the above code, we see that,
\[ {\mathsf C}(Y_{N-1},Y_N) = \arg \max_{y \in \Omega_Y} \{ Pr(\Omega_{y}) : Y_{N-1} \leq y < Y_N\} \]
For this reason, we call OSA as the maximal probability allocation code (MPA).
The following theorem from \cite{cohen-etal-ternaryfeedback} identifies the exact threshold of OSA for a given range $(y_{min},y_{max}]$.
\begin{theorem}
Given $N$ users and thresholds $(y_{min}, y_{max}]$ (i.e., $y_{min} \leq Y_{N-1} \leq Y_N \leq y_{max}$), the $y$ that maximizes the probability of success in the interval $(y_{min}, y_{max}]$ is the unique stationary point of $(y_{max} - y) (y^{N-1} - y_{min}^{N-1})$.
\end{theorem}

\textbf{Remarks}
\begin{enumerate}
\item For any $N$, and with $y_{min} = 0, y_{max} = 1$, the above expression becomes, $(1 - y)(y^{N-1}) = y^{N-1} - y^N$. The expression is maximized at $y = 1 - \frac{1}{N}$. Hence, for any $N$, $y_1 = 1 - \frac{1}{N}$.
\item As an example, for $N = 2$, repeating the above procedure will yield us $y_1 = 0.5, y_2 = 0.75, y_3 = 0.25, y_4 = 0.875, \cdots$. Note that the above values are in fact the thresholds reported in Table~\ref{tab:osa_two_users}.
    \item In Remark~\ref{rem:osa_2}, for the $N = 2$ case, we noticed that the feedback from the base station corresponding to a threshold can be viewed as the binary representation of the threshold itself. For general $N$, the feedback from the base station can still be viewed as the binary representation of the threshold, however, with the weights corresponding to a position computed from the thresholds $\{ y_k \}$ obtained from the pseudo-code. For example, the weight of the first position will be equal to $y_1$.
\end{enumerate}

\begin{center}
\begin{figure}
\includegraphics[scale=0.55]{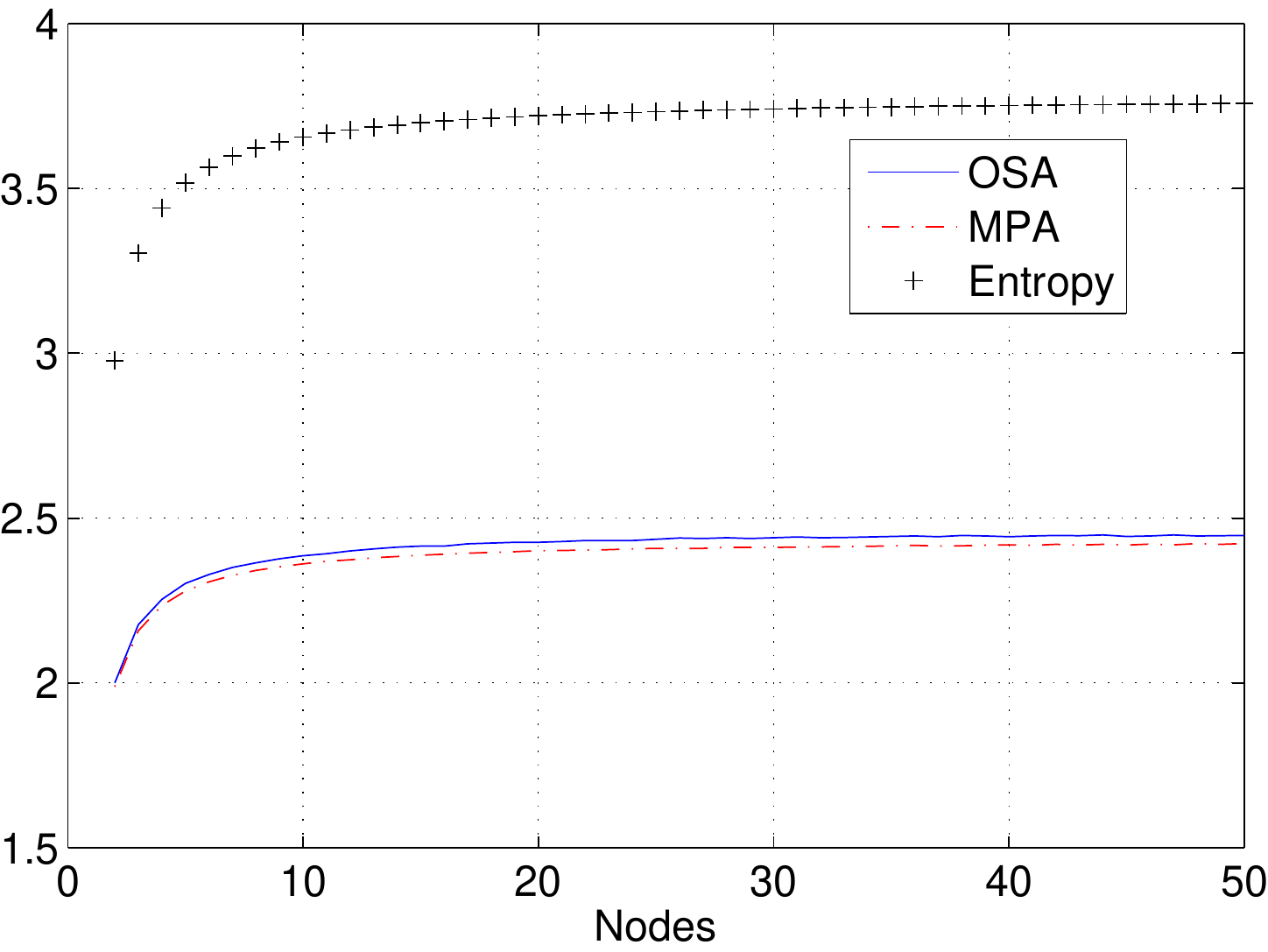}
\caption{Plot of the average delay performance of OSA and MPA as a function of the number of users $N$. We have also plotted the entropy of the threshold random variable corresponding to MPA.}
\label{fig:osa_vs_mpa}
\end{figure}
\end{center}

In Figure~\ref{fig:osa_vs_mpa}, we plot the average delay performance of OSA (as described in Section~\ref{sec:osa}) and the maximal probability allocation code. As expected, the performance of OSA and MPA are similar and in fact, MPA performs better than OSA as it identifies the optimal thresholds without any approximations (see Remark~\ref{rem:osa}).
We have also plotted in the Figure~\ref{fig:osa_vs_mpa}, the entropy of the maximal probability allocation code in bits. As expected, the entropy of the random variable reflects the average delay performance of the contention resolution algorithm as a function of $N$ very well.

\subsection{Entropy Minimization}
Entropy is a concave function of the distribution. The maximal probability allocation code identifies a local minima in the space of probability distributions. From limited numerical work (not reported in this paper), we conjecture that the maximal probability allocation code is a globally entropy minimizing strategy as well. The following theorem proves the optimality of MPA for the $N = 2$ case.

\begin{theorem}
MPA is a delay minimizing strategy and an entropy minimizing strategy for $N = 2$ case.
\end{theorem}
\begin{proof}
Let $(x, 1]$ be the contention range in the first minislot for the delay minimizing strategy.
Conditioning on the first minislot, the optimal average delay $D$ can be written as
\[ D = 2 x (1 - x) + (1 + D) (1 - 2 x (1 - x)) \]
The optimal solution of $D$ is $2$ and is obtained at $x = \frac{1}{2}$. We note from Figure~\ref{fig:osa_vs_mpa} that the average delay of MPA is $2$ and hence, MPA is a delay optimal strategy.

Let $\{ p_i \}$ be the discrete distribution that minimizes the entropy of the code for $N = 2$ case. We will define $E(\{ p_i \}) := - \sum_i p_i \log_2(pi)$ as the minimum entropy. Then, conditioning on the first minislot (as before), we have,
\[ E(\{ p_i \}) = -2 x (1-x) \log(2 x (1-x)) + E(\{ x^2 p_i\}) + E(\{ (1-x)^2 p_i \}) \]
where,
\[ E(\{a p_i\}) = - \sum_i a p_i \log_2(a p_i) \] 
Substituting in the above expression and simplifying it, we have,
{\tiny \[ E(\{ p_i \}) = \frac{ - 2 x (1 - x) \log_2(2 x (1 - x)) - x^2 \log(x^2) - (1-x)^2 \log((1-x)^2)}{2 x (1 - x)} \] }
The above expression is minimized at $x = \frac{1}{2}$ and the minimum entropy is $3$. From Figure~\ref{fig:osa_vs_mpa}, we note that the entropy of the threshold random variable for MPA is $3$ and hence, MPA is an entropy minimizing strategy as well.
\end{proof}

\section{Two Examples}
\label{sec:two_examples}
In this section, we discuss contention resolution for two different channel scenarios, a constant channel and a correlated channel. We compare the performance of OSA/MPA with the source-coding framework and illustrate the generality of our proposed technique.

\subsection{Constant Channel}
We consider a downlink wireless channel with $3$ users. We assume that the channel gain is a constant, say $1$ unit, for all users and for all time slots. The objective of the contention resolution algorithm is to identify a user from the set of $3$ users (akin to a distributed medium access problem). Suppose that the users pick a uniform random variable, $X_i$, in $[0,1]$ independent of the other users. Then, OSA can be used to resolve contention among the $3$ users by identifying the user with the largest value of $X_i$; this is a popular strategy to apply OSA for discrete channel distributions. The average number of slots required to resolve contention using OSA is then $2.12$ slots (obtained from simulations).

The OSA, in every slot, attempts to identify a $y$ such that the probability of a unique user in the interval $(y,1]$ is maximized. Here, in this example, we note that it is more appropriate to identify a $y$ that maximizes the probability of success either in $(y,1]$ or in $[0,y)$. The following algorithm is a contention resolution strategy optimized for this problem.

\begin{algorithmic}
\State Initialize: $y_{low} = 0, y_{min} = 1 - \frac{1}{3}, y_{max} = 1$
\State Initialize: $f = 0, k = 1$
\Repeat
    \State $f = (0,1,e)$ feedback for interval $(y_{min},y_{max}]$
    \If {($f = e$)}
        \State $k = k + 1$
        \State $f = (0,1,e)$ feedback from interval $[y_{low}, y_{min})$
        \If {($f = 0$)}
            \State $y_{low} = y_{min}$
            \State $y_{min} = y_{min} + (y_{max} - y_{min}) \times (1 - \frac{1}{3})$
        \EndIf
    \ElsIf {($f=0$)}
        \State $y_{max} = y_{min}$
        \State $y_{min} = y_{min} + (y_{max} - y_{low}) \times (1 - \frac{1}{3})$
    \EndIf
    \State $k = k + 1$
\Until ($f \neq 1$)
\end{algorithmic}

Using simulations, we observe that the average effort needed to resolve contention is $1.89$ slots much less than the $2.12$ slots required by OSA. The proposed algorithm makes use of the fact that, in the event of a collision, the probability that two users are involved is significantly higher than the probability that three users are involved in the collision.

The contention resolution problem was formulated as identifying a random threshold $Y$ between $Y_1$ and $Y_2$ ($Y_1 < Y \leq Y_2$) or between $Y_2$ and $Y_3$ ($Y_2 \leq Y < Y_3$). The entropy of the proposed strategy was observed to be strictly smaller than the entropy of the maximal probability allocation scheme of OSA.

\subsection{Correlated Channel}
We consider a wireless downlink channel with $N = 2$ users. We assume that the wireless channel of the two users is correlated with the sample space, $\Omega_H = \{ (4,2), (4,6), (8,6), (8,10), (12,10), (12,14), (16,14)\}$ and with the joint probabilities $p_{H} = \{ \frac{1}{7} - 6 \epsilon, \frac{1}{7} - 5 \epsilon, \frac{1}{7} - 4 \epsilon, \frac{1}{7} - 3 \epsilon, \frac{1}{7} - 2 \epsilon, \frac{1}{7} - \epsilon, \frac{1}{7} + 21 \epsilon \}$, where $0 < \epsilon << 1$. OSA would maximize the probability of success in every minislot and hence, would consider the thresholds in the following sequence (if we restrict to integer thresholds) $15, 13, 11, 9, 7, 5$ and $3$. The average number of minislots required to resolve contention with OSA/MPA is $\frac{1}{7} (1 + 2 + \cdots + 6 + 6) \approx \frac{27}{7} \approx 4$. In general, if there are $k$ channel states, then the average number of slots required to resolve contention is approximately $\frac{k}{2}$.

Consider the following alternative strategy in resolving contention.
In the first minislot, we consider the threshold value $9$ to resolve contention. If a collision occurs in the first minislot, then the next threshold would be $13$ for the second minislot and in the event of an idle first minislot, the next threshold would be set to $7$ for the second minislot. Similarly, if there is collision in the first two minislots, then, the threshold would be set to $15$ for the third minislot and so on. If there is a unique user attempting in a minislot, the contention resolution algorithm stops. The average number of minislots required to resolve contention with this strategy is approximately $3$; in general, if there are $k$ channel states, then the average number of minislots required would be $\log(k)$. We note that, for large $k$, the above strategy is strictly optimal than the OSA. The contention resolution problem can be formulated as identifying a random threshold $Y$ such that $Y_1 \leq Y < Y_2$. Clearly, the minimum entropy for the wireless channel is approximately $\log(k)$ and is equal to the average number of minislots required to resolve contention.

The two examples clearly illustrate that a maximal probability strategy like the OSA is not optimal for all channel scenarios. Also, the source-coding technique could provide us a way to identify the optimal contention resolution strategy under general channel scenarios as well.

\section{Conclusion and Future Work}
\label{sec:conclusion}
In this paper, we have modeled contention resolution for opportunistic scheduling as a source-coding
problem. The entropy of a certain random variable is seen to approximate the average number of slots required to resolve contention.
We characterized OSA as a maximal probability allocation scheme and obtained the thresholds for contention resolution (in OSA) from its source code.

We note that MPA provides us a local optima, and we conjecture that MPA is globally optimal as well (for i.i.d. channel conditions).
We believe that the information theoretic view point can be used to develop contention
resolution algorithms for a variety of other network scenarios as well (e.g., partial network information, limited channel feedback).

\bibliographystyle{IEEEtran}

\end{document}